\newif\iffull
\fulltrue
\iffull
\documentclass[a4paper,UKenglish,cleveref, autoref,
thm-restate,authorcolumns]{lipics-v2019}
\nolinenumbers

\title{Measure-Theoretic Aspects of\\ Star-Free and Group Languages}

\titlerunning{Measure-Theoretic Aspects of Star-Free and Group Languages} 

\author{Ryoma Sin'ya}{Akita University, Akita, Japan}{ryoma@math.akita-u.ac.jp}{https://orcid.org/0000-0002-8152-998X}{JST ACT-X Grant Number JPMJAX210B.}
\author{Takao Yuyama}{ZEN University, Kanagawa, Japan}{takao\_yuyama@zen.ac.jp}{}{JSPS KAKENHI Grant Number JP20H05961.}

\authorrunning{R.\, Sin'ya} 

\Copyright{Ryoma Sin'ya} 



\category{} 

\relatedversion{} 

\supplement{}

\EventEditors{John Q. Open and Joan R. Access}
\EventNoEds{2}
\EventLongTitle{42nd Conference on Very Important Topics (CVIT 2016)}
\EventShortTitle{CVIT 2016}
\EventAcronym{CVIT}
\EventYear{2016}
\EventDate{December 24--27, 2016}
\EventLocation{Little Whinging, United Kingdom}
\EventLogo{}
\SeriesVolume{42}
\ArticleNo{23}

\usepackage{graphicx}
\usepackage{url}
\usepackage{ascmac,bm}
\usepackage{amsmath,amssymb}
\usepackage{float}
\usepackage{algorithm,algorithmic}
\usepackage{extarrows}
\usepackage{arydshln}
\usepackage{hyperref}
\usepackage{mleftright}
\usepackage{tikz, tikz-cd}
\usepackage{nicematrix}
\usepackage{mathtools}
\usepackage{centernot}

\theoremstyle{plain}

\else
\documentclass{llncs}

\title{Asymptotic Approximation\\ by Regular Languages}
\author{Ryoma Sin'ya\inst{1}}
\institute{Akita University\\
\email{ryoma@math.akita-u.ac.jp}\footnote{The author is also with RIKEN
AIP.}}

\usepackage[dvipdfmx]{graphicx}
\usepackage{url}
\usepackage{ascmac,bm}
\usepackage{amsmath,amssymb}

\usepackage{float}
\usepackage{extarrows}
\usepackage{xcolor}

\fi

\newcommand{\CA}{{\cal A}}
\newcommand{\CJ}{\mathrel{\cal J}}

\newcommand{\CR}{\mathrel{\cal R}}
\newcommand{\CH}{\mathrel{\cal H}}

\newcommand{\ie}{{\it i.e.}}
\newcommand{\eg}{{\it e.g.}}
\newcommand{\resp}{resp. }
\newcommand{\cf}{{\it cf. }}

\newcommand{\defeq}{=}

\newcommand{\fin}{\mathsf{FIN}}

\newcommand{\group}{\mathsf{G}}
\newcommand{\gnil}{\mathsf{Gnil}}
\newcommand{\gsol}{\mathsf{Gsol}}
\newcommand{\gcom}{\mathsf{Gcom}}

\newcommand{\modl}{\mathsf{MOD}}
\newcommand{\pt}{\mathsf{PT}}

\renewcommand{\sf}{\mathsf{SF}}

\newcommand{\at}{\mathsf{AT}}
\newcommand{\gd}{\mathsf{GD}}

\newcommand{\ext}{\mathrm{Ext}}

\newcommand{\rext}{\mathrm{RExt}}

\newcommand{\comp}[1]{#1^{\mathrm{c}}}
\newcommand{\setmid}{\mathrel{}\middle|\mathrel{}}
\newcommand{\setdelimiter}{\setmid}
\newcommand{\set}[2]{\mleft\{\, #1 \setdelimiter #2 \,\mright\}}

\newcommand{\pd}[1]{\delta_{#1}}
\renewcommand{\d}{\pd{A}}

\newcommand{\pcd}[1]{\delta_{#1}}
\newcommand{\cd}{\pcd{A}}

\newcommand{\plm}[1]{\underline{\mu}_{#1}}
\newcommand{\pum}[1]{\overline{\mu}_{#1}}

\usepackage{mathrsfs}
\newcommand{\bool}{{\mathscr B}}
\newcommand{\monc}{{\mathscr M}}
\newcommand{\nat}{\mathbb{N}}
\newcommand{\integer}{\mathbb{Z}}
\newcommand{\modint}[1]{\integer / #1 \integer}

\newcommand{\card}[1]{\#\!\left(#1\right)}

\newcommand{\CC}{{\cal C}}
\newcommand{\CD}{{\cal D}}
\newcommand{\CB}{{\cal B}}
\newcommand{\CT}{{\cal T}}
\newcommand{\CO}{{\cal O}}

\newcommand{\reg}{\mathrm{REG}}

\begin{document}

\maketitle

\begin{abstract}
A language \(L\) is said to be \(\CC\)-measurable, where \(\CC\) is a
class of languages, if there is an infinite sequence of languages in
\(\CC\) that ``converges'' to \(L\).
We investigate the properties of \(\CC\)-measurability in the cases where \(\CC\) is \(\sf\), the class of all star-free languages, and \(\group\), the class of all group languages.
It is shown that a language \(L\) is \(\sf\)-measurable if and only if \(L\) is \(\gd\)-measurable, where \(\gd\) is the class of all generalised definite languages (a more restricted subclass of star-free languages).
This means that \(\gd\) and \(\sf\) have the same ``measuring power'', whereas \(\gd\) is a very restricted proper subclass of \(\sf\).
Moreover, we give a purely algebraic characterisation of \(\sf\)-measurable regular languages, which is a natural extension of Sch\"utzenberger's theorem stating the correspondence between star-free languages and aperiodic monoids.
We also show the probabilistic independence of star-free and group languages, which is an important application of the former result.
Finally, while the measuring power of star-free and generalised definite languages are equal, we show that the situation is rather opposite for subclasses of group languages as follows. For any two local subvarieties \(\CC \subsetneq \CD\) of group languages, we have \(\{ L \mid L \text{ is } \CC\text{-measurable}\}
\subsetneq \{ L \mid L \text{ is } \CD\text{-measurable}\}\).
\end{abstract}
 
\section{Introduction}
One of the fundamental results in semigroup theory is the Krohn-Rhodes decomposition theorem stating that every finite monoid can be decomposed in a suitable sense into a wreath product of finite aperiodic monoids and finite groups.
According to the Krohn-Rhodes decomposition theorem, the class of \emph{star-free languages} (the language counterpart of aperiodic monoids) and \emph{group languages} (languages recognised by finite groups) can be considered as two major ``building blocks'' of regular languages (see Chapter~7 of \cite{DAM} for detailed introduction of the Krohn-Rhodes theorem).
While the theory of star-free languages is well-studied (\eg, Sch\"utzenberger's algebraic characterisation~\cite{Schutzenberger1965190},  McNaughton-Papert's logical characterisation~\cite{counterfree}, the dot-depth hierarchy (described later) proposed by Cohen and Brzozowski~\cite{dot} and recent progresses relating to this hierarchy (see \cite{dot3} or the survey \cite{dotdepth}), etc.), the theory of group languages seems not, albeit that there is a result on an analogue of the dot-depth hierarchy starting with group languages~\cite{margolis} and the decidability of the so-called separation problem for group languages~\cite{place}.
Although the definition of group languages is simple from an algebraic point of view, this class does not have a language theoretic nor logical different characterisation.
To borrow the phrase used by Place--Zeitoun~\cite{place}: ``This makes it difficult to get an intuitive grasp about group languages, which may explain why this class remains poorly understood.''

This paper sheds new light on these two major building blocks --- star-free and group languages --- of regular languages, by using a measure-theoretic notion so-called \emph{\(\CC\)-measurability}.
A language \(L\) is said to be \(\CC\)-measurable if there exists an infinite sequence of pairs of languages \((K_n, M_n)_{n \in \nat}\) in \(\CC\) such that \(K_n \subseteq L \subseteq M_n\) holds for all \(n\) and the \emph{density} (described later) of the difference \(M_n \setminus K_n\) tends to zero as \(n\) tends to infinity.
The notion of \(\CC\)-measurability was introduced in \cite{S2021} and used for classifying non-regular languages by using regular languages.
The notion of \(\CC\)-measurability can be considered as a non-commutative and high-dimensional extension of Buck's measure theoretic notion for subsets of natural numbers so-called \emph{measure density}~\cite{Buck}, and the notion of \(\CC\)-measurability can be defined by using a purely measure theoretic notion called \emph{Carath{\'e}odory extension}~\cite{carathe}.
So far, the decidability and alternative characterisation of \(\CC\)-measurability is systematically studied for some subclasses of regular languages.
For example, simple and decidable characterisations of \(\pt\)-measurable and \(\at\)-measurable languages are given in \cite{ltmeasure}, where \(\pt\) is the class of all piecewise testable languages and \(\at\) is the class of all alphabet testable languages (\ie, languages defined by the first-order logic with only one variable).
For the class \(\gd\) of \emph{generalised definite} languages (\ie, languages whose membership can be checked by finitely many prefix and suffix tests), an automata theoretic characterisation of the \(\gd\)-measurable regular languages is given in \cite{gdmeasure} and it is decidable in quadratic time with respect to the size of a given deterministic automaton~\cite{ppl2024}.
Some properties of the measuring power of star-free languages is investigated in \cite{carathe}, but the decidability for regular languages was left open.
The decidability of \(\group\)-measurability, where \(\group\) is the class of all group languages, for regular languages is still unknown, while the measurability of two restricted subclasses
 \(\gcom\) (languages recognised by \emph{commutative} groups) and \(\modl\) (languages defined by length-modulo condition) of group languages are shown to be decidable for regular languages~\cite{gcommeasure}. 

\subsection*{Our contribution}
The main results of this paper (shown in Section~\ref{sec:sf}--\ref{sec:group}), relating to the measuring power of star-free and group languages, are three kinds as follows.

The first main result is about the measuring power of star-free languages (Section~\ref{sec:sf}).
It is shown that a language \(L\) is \(\sf\)-measurable if and only if \(L\) is \(\gd\)-measurable (Theorem~\ref{thm:sf}).
Since the \(\gd\)-measurability is known to be decidable for regular languages~\cite{gdmeasure}, this characterisation also gives the decidability of the \(\sf\)-measurability for regular languages.
We also give an algebraic characterisation of the \(\sf\)-measurable regular languages (Theorem~\ref{thm:algsf}), which can be considered as a natural extension of Sch\"utzenberger's theorem.

The second main result (Section~\ref{sec:ind}), which is an application of the first main result, is the \emph{probabilistic independence} of star-free and group languages (Theorem~\ref{thm:ind}) stating that
\(\cd(L \cap K) = \cd(L) \cdot \cd(K)\) holds for any
star-free language \(L\) and any group language \(K\) where
\(\cd\) is the density function defined as
\(\cd(X) \defeq \lim_{n \rightarrow \infty} \frac{1}{n} \sum_{k = 0}^{n-1} \card{X \cap A^k}/\card{A^k}\) (where \(\card{S}\) denotes the cardinality of a set \(S\)).
Since the class of star-free languages corresponds to finite aperiodic monoids (monoids having no non-trivial subgroups) via Sch\"utzenberger's theorem~\cite{Schutzenberger1965190}, it is sometimes said that star-free languages and group languages are ``orthogonal''.
For example, it is described in \cite{nato} that `This notion [aperiodic monoid] is in some sense ``orthogonal'' to the notion of groups' and in \cite{DAM} that `Aperiodic monoids are orthogonal to groups.' Our second main result gives a rigorous statement how star-free and group languages are orthogonal in the probability theoretic sense.

Between \(\gd\) and \(\sf\), there is a fine-grained infinite strict hierarchy called the \emph{dot-depth hierarchy} introduced by Cohen and Brzozowski~\cite{dot} in 1971.
For a family \(\CC\) of languages, we denote by
\(\monc \CC = \{ L_1 \cdots L_k \mid k \geq 1, L_1, \ldots, L_k \in \CC
\} \cup \{\{\varepsilon\}\}\) the monoid closure of \(\CC\),
and denote by \(\bool\CC\) the Boolean closure of \(\CC\).
The dot-depth hierarchy starts with the family \(\CB_0\) of all finite
or co-finite languages, and continues as \(\CB_{i+1} = \bool \monc \CB_{i}\) for each \(i \geq 0\).
By definition, we have \(\sf = \bigcup_{i \geq 0} \CB_i\)
and also \(\CB_0 \subsetneq \gd \subsetneq \CB_1\).
However, the first main result tells us that this infinite strict hierarchy \emph{collapses} via taking all measurable sets in the sense that 
\(\{ L \mid L \text{ is } \gd\text{-measurable}\}
 = \{ L \mid L \text{ is } \sf\text{-measurable}\}\)
holds.
The last main result (Section~\ref{sec:group}) tells us that, for the measuring power of subclasses of group languages, the situation is rather opposite.
We show that, for any two local subvarieties\footnote{A \emph{local variety} is a family of regular languages over \(A\) closed under Boolean operations and quotients.} \(\CC \subsetneq \CD \subseteq \group\) of group languages, we have \(\{ L \mid L \text{ is } \CC\text{-measurable}\}
\subsetneq \{ L \mid L \text{ is } \CD\text{-measurable}\}\), and moreover, \(\{ L \mid L \text{ is } \CC\text{-measurable}\} \cap \group = \CC\) holds (Theorem~\ref{thm:group-robust}).

In Section~\ref{sec:fut} we summarise our main results and describe some open problems.
 \section{Preliminaries}\label{sec:pre}
This section provides the precise definitions of density and measurability.
We denote by \(\reg_A\) the family of all regular languages over an alphabet \(A\). For a word \(w \in A^*\) and a letter \(a \in A\), we write \(|w|_a\) for the number of occurrences of \(a\) in \(w\). The complement of \(L\) over \(A\) is denoted by \(\comp{L} = A^* \setminus L\).
We assume that the reader has a standard knowledge of algebraic language theory (\cf \cite{Lawson,DAM,MPRI}).

For a finite monoid \(M\) and an element \(x \in M\), the unique \emph{idempotent power} (\ie, the power \(x^n\) satisfing \(x^{n} x^{n} = x^{n}\)) of \(x\) is denoted \(x^\omega\).
A monoid \(M\) is called \emph{aperiodic} if \(x^{\omega+1} = x^\omega\) holds for any \(x \in M\).
Recall that Green's \(\CH\)-relation on a monoid \(M\) is defined as \(x \CH y \Leftrightarrow (xM = yM) \land (Mx = My)\), and
\(\CJ\)-relation is defined as \(x \CJ y \Leftrightarrow MxM = MyM\).
A monoid \(M\) is said to be \({\cal G}\)-trivial if every \({\cal G}\)-class is a singleton where \({\cal G} \in \{\CH, \CJ\}\).
It is well-known that a finite monoid \(M\) is aperiodic if and only if \(M\) is \(\CH\)-trivial.
The main targets of this paper are the following three subclasses of regular languages. The \emph{star-free languages}, the \emph{generalised definite languages} \(\gd\), and the \emph{group languages} \(\group\):
\begin{align*}
 \sf_A &\defeq \set{L \subseteq A^*}{L \text{ is recognised by a finite \(\CH\)-trivial monoid}},\\
 \gd_A &\defeq \bool\set{wA^*, A^*w}{w \in A^*},\\
 \group_A &\defeq \set{L \subseteq A^*}{L \text{ is recognised by a finite group}}.
\end{align*}
Also, we denote by \(\gcom_A\) the class of all languages recognised by finite commutative groups.
A monoid \(M\) divides a monoid \(N\) if \(M\) is a monoid homomorphic image of a submonoid of \(N\).
It is well-known that a monoid \(M\) recognises a language \(L\) if and only if the syntactic monoid \(M_L\) of \(L\) divides \(M\), hence
\(L\) is in \(\sf_A\) (\resp \(\group_A\)) if and only if \(M_L\) is a finite \(\CH\)-trivial monoid (\resp a finite group).
It is easy to see that a language \(L\) is generalised definite if and only if \(L = E \cup \bigcup_{(u, v) \in F} uA^*v\) for some finite sets \(E \subseteq A^*\) and \(F \subseteq A^* \times A^*\).

\subsection{Density and measurability of formal language}
For a set $X$, we denote by $\card{X}$ the cardinality of $X$.
We denote by $\nat$ the set of natural numbers including $0$.

\begin{definition}[\cf \cite{codes}]\upshape
The \emph{density} \(\cd(L)\) of \(L \subseteq A^*\) is defined as
\[
 \cd(L) \defeq \lim_{n \rightarrow \infty} \frac{1}{n} \sum_{k = 0}^{n-1} \frac{\card{L \cap A^k}}{\card{A^k}}
\]
 if it exists, otherwise we write \(\cd(L) = \bot\).
 A language \(L\) is called \emph{null} if \(\cd(L) = 0\),
 and dually, \(L\) is called \emph{co-null} if \(\cd(L) = 1\).
\end{definition}

\begin{example}\label{ex:density}
It is known that every regular language has a rational density (\cf \cite{codes}) and the density is computable for a given automaton.
For each word \(w\), the language \(A^* w A^*\),
 the set of all words that contain \(w\) as a factor,
 is of density one (\ie, co-null). This fact follows from the so-called {\em infinite monkey theorem}:
take any word \(w\). A random word of length \(n\) contains \(w\) as a factor with probability tending to one as \(n\) tends to infinity.
\end{example}

We list some basic properties of the density as follows.
\begin{lemma}[\cite{codes}]\label{lem:density}
The function \(\cd\) is a finitely additive measure.
In particular, for any \(K, L \subseteq A^*\) with \(\cd(K) = \alpha, \cd(L) = \beta\), we have:
\begin{enumerate}
\item \(\cd(L \setminus K) = \beta - \alpha\) if \(K \subseteq L\). 
\item \(\cd(K \cup L) = \alpha + \beta\) if \(K \cap L = \varnothing\). 
\item \(\cd(uLv) = \cd(L) \cdot \card{A}^{-|uv|}\) for each \(u, v \in A^*\).
\end{enumerate}
\end{lemma}
\begin{lemma}[{\cite[Example 13.4.9]{codes}}]\label{lem:groupdensity}
 Let \(\eta\colon A^* \rightarrow G\) be a surjective morphism onto a finite group. For any \(g \in G\), we have \(\cd(\eta^{-1}(g)) = \card{G}^{-1}\).
\end{lemma}

A word \(w\) is said to be a \emph{forbidden word} of a language \(L\) if \(L \cap A^* w A^* = \emptyset\) holds.
If a language \(L\) has a forbidden word, then \(\cd(L) = 0\) holds by the infinite monkey theorem because its complement \(\overline{L} \supseteq A^* w A^*\) should be co-null.
The following lemma states that the converse of infinite monkey theorem is actually true for regular languages.
We notice that there is a context-free language with density zero but having no forbidden word (the set of all palindromes, for example).
\begin{lemma}\label{lem:zero}
Let \(L\) be a regular language.
The density of \(L\) is zero if and only if \(L\) has a forbidden word.
\end{lemma}
There are several proof strategies of Lemma~\ref{lem:zero}.
One can prove this by using the asymptotic formula of a rational series (\cf Theorem~9.8 of \cite{Salomaa:1978:ATA:578607}) with the fact that the growth rate of the language of the form \(\overline{A^* w A^*}\) is strictly smaller than \(\card{A}\)~\cite{Kim1977SomeCP}.
A more direct and automata-theoretic proof can be found in \cite{Ryoma}.
The simplest proof is the following strategy using Myhill--Nerode theorem due to Tsuboi--Takeuchi~\cite{TTmonkey}.
\begin{proof}
The ``if'' part is immediate as described.
To prove ``only if'' part by contraposition, we assume that \(L\) does not have any forbidden word as described above.
In this case, we have the equation \(A^* = \bigcup_{(u,v) \in A^* \times A^*} u^{-1} L v^{-1}\) since any word can be obtained by taking a quotient of some word in \(L\).
By Myhill--Nerode theorem, there is a finite set of pairs of words
\(F \subseteq A^* \times A^*\) such that
\(A^* = \bigcup_{(u,v) \in A^* \times A^*} u^{-1} L v^{-1} = \bigcup_{(u,v) \in F} u^{-1} L v^{-1}\).
If the language \(u^{-1} L v^{-1}\) is null for any \((u,v) \in F\), we can deduce that its finite union \(\bigcup_{(u,v) \in F} u^{-1} L v^{-1}\) must be null since \(\cd\) is a finitely additive measure.
 But this contradict with \(\cd(A^*) = \cd\left(\bigcup_{(u,v) \in F} u^{-1} L v^{-1}\right) = 1\).
Hence \(u^{-1} L v^{-1}\) should have a non-zero density for some \((u,v) \in F\), which implies that \(L\) should have a non-zero density, too.
\end{proof}

\begin{remark}\label{rem:kernel}
Let \(L\) be a regular language, \(M\) be a finite monoid and 
\(\eta\colon A^* \rightarrow M\) be a surjective morphism such that \(L = \eta^{-1}(\eta(L))\).
The \emph{kernel} \(K\) of \(M\) is a unique minimal ideal of \(M\) (it always exists since \(M\) is finite), \ie, the minimal \(\CJ\)-class of \(M\).
For every \(m\) not in the kernel \(K\) of \(M\),
\(\eta^{-1}(m) \cap A^* w A^* = \emptyset\) holds for any \(w \in \eta^{-1}(K)\) hence \(\d(\eta^{-1}(m)) = 0\) by the infinite monkey theorem.
This implies that \(\d(\eta^{-1}(K)) = 1\).
Moreover, \(\d(\eta^{-1}(K)) = 1\) implies that there exist at least one element \(n \in K\) such that \(\d(\eta^{-1}(n)) > 0\) since \(K\) is finite.
Because \(K\) is the kernel of \(M\), for any \(m \in K\) there exists a pair of 
element \(x, y \in M\) such that \(xny = m\), hence \(\eta^{-1}(m)\) contains \(u \eta^{-1}(n) v\) where \(u \in \eta^{-1}(x)\) and \(v \in \eta^{-1}(y)\) so that
\(\d(\eta^{-1}(m)) \geq \d(\eta^{-1}(n)) \cdot \card{A}^{-(|u|+|v|)}\) by Lemma~\ref{lem:density}, \ie, \(\d(\eta^{-1}(m)) > 0\) holds for each \(m \in K\).
\end{remark}
For more detailed properties of the density function \(\cd\), see Chapter~13 of \cite{codes}.

The notion of ``measurability'' on formal languages is defined by a standard measure theoretic way as follows.
\begin{definition}[\cite{S2021}]\upshape
Let \(\CC_A\) be a family of languages over \(A\).
For a language \(L \subseteq A^*\), we define its
 \emph{\(\CC_A\)-inner-density} \(\plm{\CC_A}(L)\)
 and \emph{\(\CC_A\)-outer-density} \(\pum{\CC_A}(L)\)
 over \(A\) as
 \[
 \begin{alignedat}{2}
  \plm{\CC_A}(L) &\defeq& \, \sup &\set{\cd(K)}{K \subseteq L, K \in \CC_A, \cd(K) \neq \bot} \text{ and }\\
  \pum{\CC_A}(L) &\defeq& \inf &\set{\cd(K)}{L \subseteq K, K \in \CC_A, \cd(K) \neq \bot} \text{, respectively.}
 \end{alignedat}
 \]
A language \(L\) is said to be
 \emph{\(\CC_A\)-measurable} if
 \(\plm{\CC_A}(L) = \pum{\CC_A}(L)\) holds.
 We say that an infinite sequence \((L_n)_n\) of languages over \(A\)
 {\em converges to \(L\) from inner ({\em \resp}from outer)} if
 \(L_n \subseteq L\) (\resp \(L_n \supseteq L\)) for each \(n\) and \(\lim_{n \rightarrow \infty}\cd(L_n) = \cd(L)\).
A language \(L\) is \(\CC_A\)-measurable if and only if
there exist two convergent sequences of languages in \(\CC_A\) to \(L\) from inner and outer, respectively.
\end{definition}

The following is an example of \(\gcom\)-measurable non-regular languages.
For more detailed examples of \(\CC\)-measurable/immeasurable languages for a subclass \(\CC\) of regular languages, see \cite{S2021,carathe,gcommeasure}.

\begin{example}[\cite{S2021}]\label{ex:dyck2}
 The language \(E = \{ w \in A^* \mid |w|_a = |w|_b \}\) over \(A = \{a,b\}\) is \(\gcom\)-measurable.
 For each \(k \geq 2\), the language \(L_k = \{ w \in A^* \mid |w|_a \equiv |w|_b \,\, (\!\!\!\mod k) \}\) is in \(\gcom\), because 
 \(\eta^{-1}(0) = L_k\) holds for the morphism \(\eta\colon A^* \rightarrow \modint{k}\) where \(\eta(a) = 1\) and \(\eta(b) = k-1\).
 Obviously, \(E \subseteq L_k\) holds and it follows from Lemma~\ref{lem:groupdensity} that \(\cd(L_k) = 1/k\) holds. Hence \(\cd(L_k)\) tends to zero if \(k\) tends to infinity.
\end{example}

\subsection{Local varieties and profinite equations}
For a family \(\CC_A\) of languages over \(A\), we denote by
\(\ext(\CC_A)\) (\resp \(\rext(\CC_A)\)) the class of all \(\CC_A\)-measurable languages (\resp \(\CC_A\)-measurable \emph{regular} languages) over \(A\).
A family of regular languages over \(A\) is called a \emph{local variety} (\cf \cite{Adamek}) over \(A\) if it is closed under 
 Boolean operations and left-and-right quotients.
All classes \(\gd_A, \sf_A, \gcom_A\), and \(\group_A\) are local varieties.
The notion of \(\CC_A\)-measurability is well-behaved on Boolean operations and quotients as the following lemma says.
\begin{lemma}[\cite{carathe}]\label{lem:closure}
The operator \(\ext\) is a closure,
 \ie, it satisfies the following three properties for each \(\CC_A \subseteq \CD_A \subseteq 2^{A^*}\):
 (extensive)  \(\CC_A \subseteq \ext(\CC_A)\),
 (monotone) \(\ext(\CC_A) \subseteq \ext(\CD_A)\) and
 (idempotent) \(\ext(\ext(\CC_A)) = \ext(\CC_A)\).
Moreover, 
if \(\CC_A\) is a local variety, then \(\ext(\CC_A)\) is also a local variety.
\end{lemma}

Let \(\fin_A\) be the class of all finite and co-finite languages over \(A\).
Hereafter, we only consider an alphabet \(A\) such that \(\card{A} \geq 2\) (because \(\sf_A = \gd_A = \fin_A\) and \(\group_A = \gcom_A\) holds if \(\card{A} = 1\)) and sometimes omit the subscript \(A\) for denoting these classes of languages.

The notion of local varieties and local pseudovarieties can be captured by an appropriate notion of ``equation'' called \emph{profinite equations} (\cf \cite{profinite}).
For a pair of words \(u\) and \(v\), we can define a distance between them by
\(d(u, v) = 2^{-r(u,v)}\)
where \(r(u,v) = \min\{ \lvert\CA\rvert \mid \CA \text{ is a deterministic automaton accepting one of } \{u, v\} \text{ but not the other} \}\).
Intuitively, \(r(u,v)\) represents the minimum number of states for separating \(u\) and \(v\) by using deterministic automata.
Then this function \(d\) is an ultrametric on \(A^*\), and the completion of the metric space \((A^*, d)\), denoted by \(\widehat{A^*}\), is called the set of profinite words over \(A\).
An element of \(\widehat{A^*}\) is a Cauchy sequence of (usual finite) words over \(A\) with respect to \(d\).
For example, the profinite word \(w^\omega\), associated with every word \(w\), can be defined as a limit of the sequence of words \((w^{n!})_{n \geq 1}\).
For any finite monoid \(M\) and morphism \(\eta \colon A^* \rightarrow M\), it is easy to see that \(\eta(w^{n!})\) is the idempotent power \(\eta(w)^{\omega}\) of \(\eta(w)\) for sufficiently large \(n\).
Hence, for a finite monoid \(M\), the equation \(w^{\omega+1} = w^{\omega}\) can be interpreted as \(x^{\omega+1} = x^\omega\) holds for any \(x \in M\), \ie, 
\(M\) is aperiodic (= \(\CH\)-trivial).
It is known that any local variety of languages can be captured by using profinite equations and vice versa.
We use the following theorem to prove Theorem~\ref{thm:group-robust} in Section~\ref{sec:group}.
\begin{theorem}[{\cite[Proposition 7.4]{Duality2008}}]\label{thm:duality}
  A set of regular languages over \(A\) is a Boolean quotienting algebra (\ie, local variety)
  if and only if it can be defined by a set of semigroup equations of the form \(u = v\), where \(u, v \in \widehat{A^*}\).
\end{theorem}
For any finite monoid \(M\) and morphism \(\eta\colon A^* \to M\),
we can associate its extension \(\hat{\eta}\) defined as \(\hat{\eta}((w_n)_{n \in \nat}) = \lim_{n \to \infty} \eta(w_n)\) where \((w_n)_{n \in \nat}\) is a Cauchy sequence in \(A^*\).
It is known that this \(\hat{\eta}\) is a unique \emph{uniformly continuous extension} of \(\eta\) (for details, see \eg, \cite[Chapter X, Proposition 2.10]{MPRI}).
A regular language \(L \subseteq A^*\) \emph{satisfies} an equation \(u = v\), where \(u, v \in \widehat{A^*}\),
if its syntactic morphism \(\eta\colon A^* \to M_L\) satisfies \(\hat{\eta}(u) = \hat{\eta}(v)\).

\subsection{Local varieties of finite nilpotent/solvable groups}
Here we introduce two kinds of restricted subclasses of group languages called \emph{nilpotent} and \emph{solvable} group languages.
These two subclasses of groups are natural extensions of the class of commutative groups and well-inverstigated in classical group theory (\cf \cite{Dummit-Foote-2004}).

\begin{definition}
  For a positive integer \(n\), define
  \begin{align*}
    \gnil_A^{\leq n} &\defeq \set{L \subseteq A^*}{\text{\(L\) is recognised by a finite nilpotent group of class \(\leq n\)}},\\
    \gsol_A^{\leq n} &\defeq \set{L \subseteq A^*}{\text{\(L\) is recognised by a finite solvable group of derived length \(\leq n\)}}.
  \end{align*}
\end{definition}

Recall that a group \(G\) is called a nilpotent group of class \(\leq n\) if and only if for any \(g_0, g_1, g_2, \dotsc, g_n \in G\),
the equation \([[\dotsm[[g_0, g_1], g_2]\dotsm], g_n] = 1_G\) holds,
where \([g, h] \defeq g h g^{-1} h^{-1}\) is the commutator of \(g\) and \(h\).
Hence, for example, \(\gnil_A^{\leq 2}\) is a local variety defined by the set of equations of the form
\((u v u^{\omega - 1} v^{\omega - 1}) w (u v u^{\omega - 1} v^{\omega - 1})^{\omega - 1} w^{\omega - 1} = \varepsilon\) for \(u, v, w \in A^*\).
Similarly, each \(\gnil_A^{\leq n}\) is also a local variety.
It is also known that each \(\gsol_A^{\leq n}\) is defined by a set of equations,
therefore \(\gsol_A^{\leq n}\) is also a local variety.

We notice that \(\gnil_A^{\leq 1} = \gsol_A^{\leq 1} = \gcom_A\).
An important fact is that these two kinds of subclasses form infinite strict hierarchies of group languages stated as follows (see Appendix for the detailed proof).
This fact is a key to prove Corollary~\ref{cor:strict-nilp-solv} in Section~\ref{sec:group}.
\begin{proposition}\label{prop:strict}
  Suppose \(\card{A} \geq 2\).
  For every \(n \geq 1\), we have \(\gnil_A^{\leq n} \subsetneq \gnil_A^{\leq n + 1}\) and \(\gsol_A^{\leq n} \subsetneq \gsol_A^{\leq n + 1}\).
\end{proposition}

 \section{Measuring Power of Star-Free Languages}\label{sec:sf}
In this section, we give a proof of the first main theorem, which gives a decidable characterisation of \(\sf\)-measurability for regular languages.
\begin{theorem}[S.-Y.]\label{thm:sf}
The class of star-free languages and the class of generalised definite languages have the same measuring power, \ie, \(\ext(\sf) = \ext(\gd)\).
\end{theorem}

To prove Theorem~\ref{thm:sf}, we require some lemmata, where Lemma~\ref{lem:kernelHtrivial} is the key of our proof.
\begin{lemma}[\cite{gdmeasure}]\label{lem:zeroone}
Every regular language whose density is zero or one is \(\gd\)-measurable.
\end{lemma}
\begin{lemma}\label{lem:fiberinner}
Let \(\CC \subseteq \reg_A\) be a class of regular languages closed under Boolean operations, \(M\) be a finite monoid, \(\alpha\colon A^* \rightarrow M\) be a monoid morphism, and \(L \subseteq A^*\) be a language such that \(L = \alpha^{-1}(\alpha(L))\).
If there is a convergent sequence of languages in \(\CC\) to \(\alpha^{-1}(m)\) from inner for every \(m \in M\), then \(L\) is \(\CC\)-measurable.
\end{lemma}
\begin{proof}
For each \(m \in M\), let \((L_{m,n})_n\) be a convergent sequence of languages in \(\CC\) to \(\alpha^{-1}(m)\) from inner.
Since \(\CC\) is closed under finite union,
the sequence \((\bigcup_{m \in \alpha(L)} L_{m,n})_n\) is a convergent sequence of languages in \(\CC\) to \(L\) from inner.
Since \(\CC\) is closed under finite intersection and complementation,
the sequence \((\bigcap_{m \not\in \alpha(L)} \comp{L_{m,n}})_n\) is a convergent sequence of languages in \(\CC\) to \(L\) from outer.
\end{proof}
\begin{lemma}\label{lem:kernelHtrivial}
Let \(L\) be a regular language and \(M = M_L\) be its syntactic monoid.
If every \(\CH\)-class in the kernel of \(M\) is trivial, then
\(L\) is \(\gd\)-measurable.
\end{lemma}
\begin{proof}
Let \(\eta\colon A^* \rightarrow M\) be the syntactic morphism of \(L\).
Let \(K\) be the kernel of \(M\).
If \(m \not\in K\), then the infinite monkey theorem implies 
\(\d(\eta^{-1}(m)) = 0\) hence \(\eta^{-1}(m)\) is \(\gd\)-measurable by Lemma~\ref{lem:zeroone}.
Let \(m \in M\) be an element in \(K\).
In this case, we have \(\CH_m = mM \cap Mm\).
Since \(K\) is \(\CH\)-trivial by the assumption, this implies
\(\CH_m = mM \cap Mm = \{m\}\).
For the language \(L_m \defeq \eta^{-1}(m)\), we have
\[
 L_m = \eta^{-1}(mM \cap Mm)
 = \eta^{-1}(mM) \cap \eta^{-1}(Mm).
\]
Let \(w \in \eta^{-1}(K)\).
If a word \(u \in A^*\) contains \(w\) as a factor (\ie, \(u \in A^* w A^*\)), then \(\eta(u) \in K\) and \(\eta(uA^*) = \eta(u) M\) is a minimal right ideal. 
Notice that the set of all minimal right ideals is a partition of the kernel 
\(
 K = \bigsqcup_{n \in K / \CR} nM
\) where \(\bigsqcup\) denotes the disjoint union.
For each \(\ell \geq 0\), we define a generalised definite language \(L_\ell\) as
\(
 L_\ell = \bigsqcup\set{uA^*}{u \in A^\ell, \eta(uA^*) \subseteq mM}.
\)
The inclusion \(L_\ell \subseteq \eta^{-1}(mM)\) is obvious.
We can show that \(L_\ell\) converges to \(\eta^{-1}(mM)\) from inner as follows.
\begin{align*}
& \,\, \d\left(\eta^{-1}(mM) \setminus L_\ell\right)\\
= & \,\,
\d(\eta^{-1}(mM)) - \d\left(\bigsqcup\{ uA^* \mid u \in A^\ell, \eta(uA^*) \subseteq mM\}\right)\\
\leq & \,\,
\sum_{n \in K / \CR} \left( \d(\eta^{-1}(nM)) - \d\left(\bigsqcup\{ uA^* \mid u \in A^\ell, \eta(uA^*) \subseteq nM\}\right) \right)\\
= & \,\,
\sum_{n \in K / \CR} \d(\eta^{-1}(nM))
 - \sum_{n \in K / \CR} \d\left(\bigsqcup\{ uA^* \mid u \in A^\ell, \eta(uA^*) \subseteq nM\}\right)\\
= & \,\,
\d(\eta^{-1}(K))
 - \d\left(\bigsqcup\{ uA^* \mid u \in A^\ell, \eta(uA^*) \subseteq nM \text{ for some } n \in K/\CR\}\right)\\
\leq & \,\,
1 - \d\left(\bigsqcup\{ uA^* \mid u \in A^\ell \cap A^* w A^*\}\right)
\end{align*}
where the last term tends to zero as \(\ell\) tends to infinity, since 
\[
\d\left(\bigsqcup\{ uA^* \mid u \in A^\ell \cap A^* w A^*\}\right)
= \frac{\card{A^\ell \cap A^* w A^*}}{\card{A^\ell}}
\]
tends to one as \(\ell\) tends to infinity by the infinite monkey theorem.
By the same argument, there is a convergent sequence \((L'_\ell)_{\ell \geq 0}\) of generalised definite languages to \(\eta^{-1}(Mm)\) from inner.
This means that the sequence \((L_\ell \cap L'_\ell)_{\ell \geq 0}\)
of generalised definite languages converges to \(\eta^{-1}(mM) \cap \eta^{-1}(Mm) = L_m\) from inner.
We showed that, for every \(m \in M\),
there is a convergent sequence of generalised definite languages to \(L_m = \eta^{-1}(m)\) from inner.
Hence \(L\) is \(\gd\)-measurable by Lemma~\ref{lem:fiberinner} because \(\gd\) is closed under Boolean operations.
\end{proof}

\begin{proof}[Proof of Theorem~\ref{thm:sf}]
Since we have \(\gd \subseteq \sf\), it suffices to show that
\(\sf \subseteq \ext(\gd)\), which implies 
\(\ext(\gd) \subseteq \ext(\sf) \subseteq \ext(\ext(\gd)) = \ext(\gd)\) by Lemma~\ref{lem:closure}.
The syntactic monoid of any star-free language is \(\CH\)-trivial, hence \(\sf \subseteq \ext(\gd)\) holds by Lemma~\ref{lem:kernelHtrivial}.
\end{proof}

Actually, the converse direction of Lemma~\ref{lem:kernelHtrivial} holds. The following theorem gives an algebraic characterisation of \(\gd\)-measurable (\(= \) \(\sf\)-measurable) regular languages.
This theorem can be considered as a natural extension of Sch\"utzenberger's theorem stating that \(L\) is star-free if and only if its syntactic monoid is \(\CH\)-trivial.
\begin{theorem}[S.]\label{thm:algsf}
Let \(L\) be a regular language and \(M = M_L\) be its syntactic monoid.
Then \(L\) is \(\gd\)-measurable if and only if the kernel of \(M\) is \(\CH\)-trivial, \ie, every \(\CH\)-class in the kernel of \(M\) is trivial.
\end{theorem}
\begin{proof}
We show the ``only if'' part by contraposition.
Let \(K\) be the kernel of \(M\).
Suppose that some \(\CH\)-class in the kernel of \(M\) is not trivial.
This implies that, there is some \(k \geq 2\) such that \(\card{\CH_m} = k\) for every \(m \in K\) since every \(\CH\)-class in \(K\) is a group and all these groups are isomorphic (\cf Chapter V, Section 4.1 of \cite{MPRI}).
If \(\d(L) = 0\), then \(K\) should be trivial by Lemma~\ref{lem:zero} (the syntactic image of any forbidden word is \(0\) the zero element and hence \(K = \{0\}\)).
Hence we may assume that \(\d(L) > 0\).
Let \(\eta\colon A^* \rightarrow M\) be the syntactic morphism of \(L\) and let \(P \defeq \eta(L)\).
Because \(\d(L) > 0\) and the infinite monkey theorem, \(P\) and \(K\) have a non-empty intersection and hence we can pick some element \(m \in P \cap K\).
By the assumption there is some \(n \in \CH_m\) such that \(m \neq n\).
Since \(M\) is the syntactic monoid of \(L\), it is known that 
\(P\) is \emph{disjunctive subset} of \(M\) (\cf Proposition 10.3.3 of \cite{Lawson}), \ie,
for each \(a, b \in M\), if 
\(cad \in P \Leftrightarrow cbd \in P\) for any pair \(c,d \in M\), then \(a = b\).
Hence \(m \neq n\) implies \(m' = cmd \in P\) but \(n' = cnd \notin P\) for some \(c, d \in M\).
Now we show that \(u A^* v \subseteq L\) implies
\(u A^* v \cap \eta^{-1}(m') = \emptyset\) for any pair of words \(u\) and \(v\).
This concludes that, every generalised definite subset \(K\) of \(L\) satisfies \(\d(K) \leq \d(L) - \d(\eta^{-1}(m')) < \d(L)\) (notice that \(\d(\eta^{-1}(m')) > 0\) as stated in Remark~\ref{rem:kernel}), \ie, \(L\) is \(\gd\)-immeasurable.

Assume \(u A^* v \cap \eta^{-1}(m') \neq \emptyset \) holds for some pair of words \(u\) and \(v\).
Then \(\eta(uwv) = m'\) for some \(w \in A^*\).
Since \(m' \mathrel{\CH} n'\), there is some \(x \in M\) such that \(xm' = n'\).
Let \(w_x \in \eta^{-1}(x)\).
Then the word \(u (wvu)^{k-1} w v w_x uwv \in u A^* v\) satisfies
\[
 \eta(u (wvu)^{k-1} w v w_x uwv) = 
\eta(uwv)^k \eta(w_x) \eta(uwv) =
(m')^k x m' = n'.
\]
Hence \(u A^* v \cap \eta^{-1}(n') \neq \emptyset\), which implies \(u A^* v \not\subseteq L\) since \(\eta^{-1}(n') \cap L = \emptyset\).
\end{proof}
 \section{Probabilistic Independence of Star-Free and Group Languages}\label{sec:ind}
In probability theory, two events \(A\) and \(B\) are said to be independent
if their joint probability equals the product of their probabilities:
\(P(A \cap B) = P(A) \cdot P(B)\).
One important application of Theorem~\ref{thm:sf} is the following probabilistic independence of star-free and group languages, which emphasises how these two language classes are orthogonal.
\begin{theorem}[Y.-S.]\label{thm:ind}
For any star-free language \(L\) and any group language \(K\),
we have 
\[
 \cd(L \cap K) = \cd(L) \cdot \cd(K).
\]
\end{theorem}
The probabilistic independence of generalised definite languages and group languages is proved more easily; we use it as a lemma for proving Theorem~\ref{thm:ind}.
\begin{lemma}\label{lem:ind}
For any generalised definite language \(L\) of the form \(u A^* v\) for some \(u, v \in A^*\) and any group language \(K\),
we have 
\[
 \cd(L \cap K) = \cd(L) \cdot \cd(K).
\]
\end{lemma}
\begin{proof}
Let \(L = u A^* v\) be a generalised definite language, and \(K\) be a group language such that \(K = \eta^{-1}(S)\) for some finite group \(G\), where \(S\) is a subset of \(G\) and \(\eta\colon A^* \rightarrow G\) is a surjective morphism.
Observe that \(L \cap K = u (u^{-1} K v^{-1}) v\) holds.
Since we have \(u^{-1} K v^{-1} = u^{-1} \eta^{-1}(S) v^{-1} 
= \eta^{-1}( \eta(u)^{-1} S \eta(v)^{-1})\),
we can conclude that 
\(\cd(L \cap K) = \cd(u(\eta^{-1}( \eta(u)^{-1} S \eta(v)^{-1})) v) =
\card{A}^{-|uv|} \cdot \cd(\eta^{-1}( \eta(u)^{-1} S \eta(v)^{-1}))
= \cd(L) \cdot \cd(K)
\)
by Lemma~\ref{lem:density} and Lemma~\ref{lem:groupdensity} since 
\(\card{S} = \card{\eta(u)^{-1} S \eta(v)^{-1}}\).
\end{proof}
\begin{proof}[Proof of Theorem~\ref{thm:ind}]
Let \(L\) be a star-free language and \(K\) be a group language.
We only show that \(\cd(L \cap K) \geq \cd(L) \cdot \cd(K)\) holds;
the reverse inequality \(\cd(L \cap K) \leq \cd(L) \cdot \cd(K)\) is also proved in the same way.
By Theorem~\ref{thm:sf}, there is a sequence \(M_n = (\bigcup_{(u, v) \in F_n} u A^* v)_{n \geq 0}\) of generalised definite languages convergent to \(L\) from inner, where \(F_n \subseteq A^* \times A^*\) is finite for each \(n\).
Without loss of generality, we may assume that \(uA^*v \cap u'A^*v' = \emptyset\) for each different pairs \((u,v), (u',v') \in F\).
Then we have
\[
 \cd(L \cap K) \geq \cd\left(\bigcup_{(u,v) \in F_n} u A^* v \cap K\right)
 = \sum_{(u,v) \in F_n} \cd(u A^* v) \cdot \cd(K)
 = \cd(M_n) \cdot \cd(K)
\]
by Lemma~\ref{lem:ind} where \(\cd(M_n)\) tends to \(\cd(L)\) as \(n\) tends to infinity.
Thus \(\cd(L \cap K) \geq \cd(L) \cdot \cd(K)\).
\end{proof}

 \section{Hierarchies of Group Languages Are Not Collapsed by \(\rext\)}\label{sec:group}

The main result of this section is Corollary~\ref{cor:strict-nilp-solv},
which is a consequence of Theorem~\ref{thm:group-robust} below and states that
the classes of regular languages approximable by nilpotent/solvable group languages form strict hierarchies.
This is completely different from the situation in Theorem~\ref{thm:sf},
which shows that the dot-depth hierarchy is collapsed by \(\rext\).

\begin{theorem}[Y.]\label{thm:group-robust}
  If \(\CC \subseteq \group_A\) is a local variety of group languages, then \(\rext(\CC) \cap \group_A = \CC\).
  In particular, if \(\CC \subsetneq \CD \subseteq \group_A\) are local subvarieties,
  then \(\rext(\CC) \cap \CD = \CC\), \(\CD \not\subseteq \rext(\CC)\), and \(\rext(\CC) \subsetneq \rext(\CD)\).
\end{theorem}
Our previous result \cite[Theorem 1]{gcommeasure} is the special case of Theorem~\ref{thm:group-robust} where \(\CC = \gcom_A\).

\begin{proof}
 Let \(L \in \group_A \setminus \CC\).
  It suffices to show that \(L\) is \(\CC\)-immeasurable.
  Since \(\CC\) is a local variety and \(L \notin \CC\), Theorem~\ref{thm:duality} implies
  that there exists an equation \(u = v\) (where \(u, v \in \widehat{A^*}\)) that every language in \(\CC\) satisfies but \(L\) does not.
  Let \(G\) be the syntactic monoid of \(L\), which is a group, and \(\eta\colon A^* \to G\) be the syntactic morphism of \(L\).

  We claim that there exists \(w \in \widehat{A^*}\) such that the equation \(w = \varepsilon\) is equivalent to \(u = v\) on any \(A\)-generated finite group.
  Indeed, if \((u_n)_{n \in \nat}\) and \((v_n)_{n \in \nat}\) are Cauchy sequences in \(A^*\) converging to \(u\) and \(v\) in \(\widehat{A^*}\) respectively,
  then the Cauchy sequence \((u_n v_n^{n! - 1})_{n \in \nat}\) satisfies the claim.

  Hence, for every \(L_1 \in \CC\) and its syntactic morphism \(\theta_1\colon A^* \to G_1\),
  we have \(\hat{\theta}_1(w) = \hat{\theta}_1(\varepsilon) = 1_{G_1}\) and \(\hat{\eta}(w) \neq \hat{\eta}(\varepsilon) = 1_G\).
  Since \(\hat{\eta}(w) \neq 1_G\) and \(G\) is the syntactic monoid of \(L\), there exist \(x, y \in A^*\)
  such that
  \begin{equation}
    \eta(x) \hat{\eta}(w) \eta(y) \in \eta(L) \centernot\iff \eta(x y) \in \eta(L). \label{eq:ebc45b27}
  \end{equation}
  Since the two languages \(\eta^{-1}(\eta(x)), \eta^{-1}(\eta(y)) \subseteq A^*\) are group languages,
  they are non-null by Lemma~\ref{lem:groupdensity}.
 In particular, there should be some word \(w \in \eta^{-1}(\eta(x))\), and hence \(\eta^{-1}(\eta(x)) \eta^{-1}(\eta(y)) \supseteq w \eta^{-1}(\eta(y))\) should have a positive density by Lemma~\ref{lem:density}.
  Hence the concatenation \(I \defeq \eta^{-1}(\eta(x)) \eta^{-1}(\eta(y))\) is non-null.
  To complete the proof, it suffices to show that \(L_1, L_2 \in \CC\) and \(L_1 \subseteq L \subseteq L_2\) implies \(I \subseteq L_2 \setminus L_1\).
  Let \(\theta_1\colon A^* \to G_1\) and \(\theta_2\colon A^* \to G_2\) be the syntactic morphism of \(L_1\) and \(L_2\), respectively.
  Let \((w_n)_{n \in \nat}\) be a Cauchy sequence in \(A^*\) converging to \(w\) in \(\widehat{A^*}\).
  Then we have \(\eta(w_k) = \hat{\eta}(w)\), \(\theta_1(w_k) = \hat{\theta}_1(w)\), and \(\theta_2(w_k) = \hat{\theta}_2(w)\) for sufficiently large \(k \in \nat\).

  We show \(I \subseteq L_2 \setminus L_1\).
  Let \(z \in I\).
  Then there exist \(\tilde{x} \in \eta^{-1}(\eta(x))\) and \(\tilde{y} \in \eta^{-1}(\eta(y))\) such that \(z = \tilde{x} \tilde{y}\).
  Let \(i \in \{1, 2\}\).
  Since \(\theta_i(w_k) = 1_{G_i}\) and \(G_i\) is the syntactic monoid of \(L_i\),
  we have \(\theta_i(\tilde{x} w_k \tilde{y}) \in \theta_i(L_i) \iff \theta_i(\tilde{x} \tilde{y}) \in \theta_i(L_i)\),
  \ie,
  \begin{equation}
    z = \tilde{x} \tilde{y} \in L_i \iff \tilde{x} w_k \tilde{y} \in L_i. \label{eq:fbaddcb1}
  \end{equation}
  It follows from \eqref{eq:ebc45b27}, \(\eta(\tilde{x}) = \eta(x)\), \(\eta(\tilde{y}) = \eta(y)\), and \(\hat{\eta}(w) = \eta(w_k)\) that
  \begin{equation}
    z = \tilde{x} \tilde{y} \in L \centernot\iff \tilde{x} w_k \tilde{y} \in L. \label{eq:500d85c6}
  \end{equation}
We show \(z \in L_2 \setminus L_1\) by considering two cases depending on whether \(z \in L\) or \(z \notin L\).
  First, suppose that \(z \in L\).
  Since \(L \subseteq L_2\), we have \(z \in L_2\).
  On the other hand, by \eqref{eq:500d85c6}, we know that \(\tilde{x} w_k \tilde{y} \notin L \supseteq L_1\), which implies \(\tilde{x} w_k \tilde{y} \notin L_1\).
  Then, by \eqref{eq:fbaddcb1}, it follows that \(z \notin L_1\), and thus we obtain \(z \in L_2 \setminus L_1\).
  The case where \(z \notin L\) can be proved in a similar manner.
\end{proof}

\begin{corollary}\label{cor:strict-nilp-solv}
  Suppose \(\card{A} \geq 2\).
  For \(n \geq 1\), we have
  \begin{itemize}
    \item \(\rext(\gnil_A^{\leq n}) \subsetneq \rext(\gnil_A^{\leq n + 1}) \subsetneq \rext(\group_A)\), and
    \item \(\rext(\gsol_A^{\leq n}) \subsetneq \rext(\gsol_A^{\leq n + 1}) \subsetneq \rext(\group_A)\).
  \end{itemize}
\end{corollary}
\begin{proof}
  By Proposition~\ref{prop:strict} and Theorem~\ref{thm:group-robust}.
\end{proof}

 \section{Conclusion and Open Problems}\label{sec:fut}
Theorem~\ref{thm:sf} and our previous work~\cite{gdmeasure,ppl2024} give a decidable characterisation of the \(\sf\)-measurable regular languages.
Its algebraic characterisation (Theorem~\ref{thm:algsf}) can be considered as a natural extension of Sch\"utzenberger's theorem for star-free languages.
The decidability of the \(\group\)-measurability for regular languages is still open.
Perhaps the profinite method can be useful to solve this problem.
The probabilistic independence of star-free and group languages describes how these two classes are orthogonal.

In 2023, the first author posed a question whether the equation
\(\reg = \bool(\rext(\sf) \cup \rext(\group))\) holds or not, at some domestic workshops in Japan.
If this kind of equation \emph{were} true, it can be considered as a ``measure-theoretic decomposition theorem'' of regular languages.
Later, the second author claimed that a weaker equation
\(\reg = \rext(\bool(\sf \cup \group))\) holds and
this result was announced at the international conferences Highlights 2024 by the first author.
However, after some discussion between these two authors, a flaw was found in the proof of \(\reg = \rext(\bool(\sf \cup \group))\).
Actually, this equation does \emph{not} hold as follows.
\begin{proposition}\label{prop:counterexample}
Let \(\CT_4\) be the full transformation monoid on the four elements set \(\{1,2,3,4\}\) whose identity element is the identity function \(\mathrm{id}\).
Define \(e, f \in \CT_4\) as
\begin{align*}
 e\colon 1 \mapsto 1, 2 \mapsto 2, 3 \mapsto 1, 4 \mapsto 2 \quad \text{ and } \quad
 f\colon 1 \mapsto 4, 2 \mapsto 3, 3 \mapsto 3, 4 \mapsto 4.
\end{align*}
Let \(M\) be the submonoid of \(\CT_4\) generated by \(e\) and \(f\).
For an alphabet \(A = \{a,b\}\), define a morphism \(\eta\colon A^* \rightarrow M\) as \(\eta(a) = e\) and \(\eta(b) = f\).
Then \(\eta^{-1}(e)\) is \(\bool(\sf_A \cup \group_A)\)-immeasurable.
In particular, we have \(\reg_A \supsetneq \rext_A(\bool(\sf_A \cup \group_A))\).
\end{proposition}
\begin{proof}
By simple calculation, one can easily observe that 
\(e^2 = e, f^2 = f, efefe = e, fefef = f\) hold so that 
\(M = \{\mathrm{id},e,ef,efe,efef,f,fe,fef,fefe\}\).
The kernel of \(M\) is \(M \setminus \{\mathrm{id}\}\) and
\(\CH_e = \{e, efe\}\) and \(\CH_f = \{f, fef\}\) are maximal subgroups of \(M\), respectively.

Thanks to Theorem~\ref{thm:sf}, we have 
\(\rext_A(\bool(\sf_A \cup \group_A)) = \rext_A(\bool(\gd_A \cup \group_A))\).
Notice again that \(\cd(\eta^{-1}(e)) > 0\) holds as stated in Remark~\ref{rem:kernel}.
Let \(L \in \gd_A\) and \(K \in \group_A\).
To show that \(\eta^{-1}(e)\) is \(\bool(\gd_A \cup \group_A)\)-immeasurable, it is enough to show that \(L \cap K \subseteq \eta^{-1}(e)\) implies \(\cd(L \cap K) = 0\) (because this condition implies there is no convergent sequence of languages in \(\bool(\gd_A \cup \group_A)\) to \(\eta^{-1}(e)\) from inner).
We actually show that \(L \cap K \subseteq \eta^{-1}(e)\) implies \(L \cap K\) is a finite set.

Assume \(L \cap K \subseteq \eta^{-1}(e)\) and assume that
\(L \cap K\) is infinite.
Since \(L \in \gd_A\), there is some constant \(N\) such that,
for any word \(u, v \in A^*\) with \(|u|, |v| > N\),
\(uwv \in L \Leftrightarrow uw'v \in L\) for any pair of words
\(w\) and \(w'\).
By the assumption, \(L \cap K \subseteq \eta^{-1}(e)\) is infinite hence there is a pair of words \(u\) and \(v\) such that \(uv \in L \cap K\), \(|u|, |v| > N\), \(\eta(u) \in eM\) and \(\eta(v) \in Me\).
Since \(K\) is a group language, there is a finite group \(G\) and a morphism \(\varphi\colon A^* \rightarrow G\) such that
\(K = \varphi^{-1}(\varphi(K))\).
Since \(\eta(u) \in eM = \{e, ef, efe, efef\}\) and \(\eta(v) \in Me = \{fefe, efe, fe, e\}\),
one can choose a word \(w = a_1 a_2 \cdots a_k \in A^*\) so that \(\eta(u) \eta(w) \eta(v) = e f e\).
(In fact, we may assume \(k \leq 2\).)
Let \(n = \card{G}\) and \(\widetilde{w} = a_1^n a_2^n \cdots a_k^n\).
Since \(\eta(a) = e\) and \(\eta(b) = f\) are both idempotents, we have
\(\eta(u \widetilde{w} v) =
\eta(u) \eta(\widetilde{w}) \eta(v) =
\eta(u) \eta(w) \eta(v) = efe
\)
thus \(u\widetilde{w}v\) is not in \(\eta^{-1}(e)\).
On the other hand, we have \(u \widetilde{w}v \in L\) since \(|u|, |v| > N\) and \(uv \in L\).
But \(\varphi(u\widetilde{w}v) = \varphi(uv)\) holds by the definition of \(\widetilde{w}\) which implies
\(u\widetilde{w}v \in \varphi^{-1}(\varphi(uv)) \subseteq K\).
This means that \(L \cap K\) contains the word \(u\widetilde{w}v\) which is not contained in \(\eta^{-1}(e)\), which contradicts with the assumption that \(L \cap K \subseteq \eta^{-1}(e)\).
\end{proof}
The authors are still interested in whether there is a suitable notion \(\CO\colon 2^{A^*} \rightarrow 2^{A^*}\) of ``operation'' on languages such that the equation like \(\reg = \CO(\rext(\sf) \cup \rext(\group))\) or \(\reg = \rext(\CO(\sf \cup \group))\) holds.
It must be richer than the Boolean closure operator \(\bool\), but might be simpler than the operation on languages corresponding to the wreath product operation on monoids.\\

\noindent
{\bf Acknowledgement.}
The authors thank to anonymous reviewers for many valuable comments for improving the presentation of this paper.
The first author gratefully acknowledge helpful and encouraging discussion with Miko{\l}aj Boja\'nczyk and Volker Diekert.
Volker pointed out an interesting further research direction on this topic.

\appendix
\section{Appendix}\label{sec:appendix}

To prove Proposition~\ref{prop:strict},
we need the following Lemmata~\ref{lem:two-gen-nilpotent} and \ref{lem:two-gen-solvable}.

\begin{lemma}\label{lem:two-gen-nilpotent}
  For \(n \geq 1\), there is a two-generated finite nilpotent group of class \(n\).
\end{lemma}
\begin{proof}
  For each \(n \geq 1\), the dihedral group \(D_{2^{n + 1}}\) of order \(2^{n + 1}\)
  is a two-generated finite nilpotent group of class \(n\) (see \eg, \cite[p.\ 191, Example]{Dummit-Foote-2004}).
\end{proof}

\begin{lemma}\label{lem:two-gen-solvable}
  For \(n \geq 1\), there is a two-generated finite solvable group of derived length \(n\).
\end{lemma}
To prove Lemma~\ref{lem:two-gen-solvable}, we use the following embedding theorem due to Neumann-Neumann,
and Lemmata~\ref{lem:solvable-length-n} and \ref{lem:solvable-reduce} below.

\begin{theorem}[{\cite[Theorem 5.6]{neumann-neumann-1959}}]\label{thm:embed-solvable}
  Every finite solvable group \(G\) of derived length \(n\) can be embedded in
  some two-generated finite solvable group \(H\) of derived length at most \(n + 2\).
\end{theorem}

\begin{lemma}[{\cite[9.23 Corollary]{rose-group}}]\label{lem:solvable-length-n}
  For each \(n \geq 1\), there exists a finite solvable group of derived length \(n\).
\end{lemma}

\begin{lemma}\label{lem:solvable-reduce}
  If \(G\) is a solvable group of derived length \(n + 1\), then some factor group of \(G\) is solvable of derived length \(n\).
\end{lemma}
\begin{proof}
  Let \(G = G^{(0)} \rhd G^{(1)} \rhd \dotsb \rhd G^{(n)} \rhd G^{(n + 1)} = \{1_G\}\) be the derived series of \(G\),
  where \(G^{(i + 1)} \defeq [G^{(i)}, G^{(i)}]\).
  Since each \(G^{(i + 1)}\) is a characteristic subgroup of \(G^{(i)}\), the group \(G^{(n)}\) is also a characteristic subgroup of \(G\),
  hence in particular a normal subgroup of each \(G^{(i)}\).
  Let \(\bar{G} \defeq G / G^{(n)}\).
  Then \(G^{(0)} / G^{(n)} \rhd G^{(1)} / G^{(n)} \rhd \dotsb \rhd G^{(n)} / G^{(n)} = \{1\}\) is the derived series of \(\bar{G}\) of length \(n\),
  since we have an isomorphism \([G^{(i)} / G^{(n)}, G^{(i)} / G^{(n)}] \cong [G^{(i)}, G^{(i)}] / G^{(n)} = G^{(i + 1)} / G^{(n)}\).
\end{proof}

\begin{proof}[Proof of Lemma~\ref{lem:two-gen-solvable}]
  By Lemma~\ref{lem:solvable-length-n}, there exists a finite solvable group \(G\) of derived length \(n\).
  By Theorem~\ref{thm:embed-solvable},
  \(G\) can be embedded in some two-generated finite group \(H\) of derived length at most \(n + 2\).
In particular, the derived length of \(H\) is one of \(n\), \(n + 1\), or \(n + 2\).
  Thus, by Lemma~\ref{lem:solvable-reduce},
  we obtain a two-generated finite solvable group of derived length \(n\) as a factor group of \(H\).
\end{proof}

Lemma~\ref{lem:WP-synt} below is the last piece of the proof of Proposition~\ref{prop:strict}.
\begin{lemma}[folklore, \eg, {\cite[3.4.9 Proposition]{Holt2017}}]\label{lem:WP-synt}
  Let \(G\) be a group and \(\eta\colon A^* \to G\) be a surjective morphism.
  Then the syntactic monoid of the language \(\eta^{-1}(1_G) \subseteq A^*\) is isomorphic to \(G\).
\end{lemma}
The language \(\eta^{-1}(1_G)\) is known as the \emph{word problem} for the group \(G\) (over \(A\)).

\begin{proof}[Proof of Proposition~\ref{prop:strict}]
  By Lemma~\ref{lem:two-gen-nilpotent}, there exists a two-generated finite nilpotent group \(G\) of class \(n + 1\).
  Since \(G\) is finite and generated by two elements as a group, it is generated by the two elements as a monoid.
  Hence there exists a surjective morphism \(\eta\colon A^* \to G\).
  Since \(\eta^{-1}(1_G)\) is clearly recognised by \(G\), we have \(\eta^{-1}(1_G) \in \gnil_A^{\leq n + 1}\).
  Contrary, suppose \(\eta^{-1}(1_G) \in \gnil_A^{\leq n}\).
  Then \(\eta^{-1}(1_G)\) is recognised by some finite nilpotent group \(H\) of class \(n\),
  hence the syntactic monoid of \(\eta^{-1}(1_G)\), which is isomorphic to \(G\) by Lemma~\ref{lem:WP-synt}, divides \(H\).
  In general, if \(K\) is a nilpotent group of class \(c\),
  then every subgroup of \(K\) and every factor group of \(K\) is nilpotent of class at most \(c\).
  Hence \(G\) is nilpotent of class at most \(n\), contradiction.
  (Note that every submonoid of a group is a subgroup.)
  Thus \(\gnil_A^{\leq n} \subsetneq \gnil_A^{\leq n + 1}\).

  By Lemma~\ref{lem:two-gen-solvable}, there exists a two-generated finite solvable group \(G\) of derived length \(n + 1\).
  Then one can similarly show \(\eta^{-1}(1_G) \in \gsol_A^{\leq n + 1} \setminus \gsol_A^{\leq n}\) for a surjective morphism \(\eta\colon A^* \to G\).
\end{proof}
 
\end{document}